\documentclass[letterpaper, 10 pt, conference]{ieeeconf}  

\IEEEoverridecommandlockouts                              

\usepackage{graphicx}
\usepackage{amsmath}
\usepackage{mathrsfs}
\usepackage{mathtools}
\DeclareMathOperator*{\argmax}{argmax}

\usepackage{yhmath}
\usepackage{amssymb}
\usepackage{cite}
\usepackage{url}
\usepackage{subcaption}

\title{\LARGE \bf
Markov Potential Game Construction and Multi-Agent Reinforcement Learning with Applications to Autonomous Driving
}

\author{Huiwen Yan and Mushuang Liu
\thanks{This work was supported by DARPA Young Faculty Award with the grant number D24AP00321.}
\thanks{Huiwen Yan and Mushuang Liu are with the Department of Mechanical Engineering at Virginia Tech, Blacksburg, VA, USA
        {\tt\small huiweny@vt.edu,  mushuang@vt.edu}}%
}

\begin{document}

\newtheorem{thm}{Theorem}
\newtheorem{remark}{Remark}
\newtheorem{lemma}{Lemma}
\newtheorem{prop}{Proposition}
\newtheorem{defn}{Definition}
\newtheorem{condi}{Condition}
\newtheorem{assump}{Assumption}

\maketitle
\thispagestyle{empty}
\pagestyle{empty}

\begin{abstract}
Markov games (MGs) provide a mathematical foundation for multi-agent reinforcement learning (MARL), enabling self-interested agents to learn their optimal policies while interacting with others in a shared environment. However, due to the complexities of an MG problem, seeking (Markov perfect) Nash equilibrium (NE) is often very challenging for a general-sum MG. Markov potential games (MPGs), which are a special class of MGs, have appealing properties such as guaranteed existence of pure NEs and guaranteed convergence of gradient play algorithms, thereby leading to desirable properties for many MARL algorithms in their NE-seeking processes. However, the question of how to construct MPGs has remained open. This paper provides sufficient conditions on the reward design and on the Markov decision process (MDP), under which an MG is an MPG. Numerical results on autonomous driving applications are reported. 

\end{abstract}

\section{INTRODUCTION}
Reinforcement learning (RL) has demonstrated  success in diverse applications, e.g., resource allocation\cite{resourceAllocation}, energy management\cite{energyManagement} and robotics\cite{robotics}. The RL problem is often modeled as an MDP\cite{mdpBasis}, where a single agent interacts with the environment to iteratively update its policy until the optimal \cite{singleAgentRL1},\cite{singleAgentRL2}. However, modern complex systems are often composed of multiple decision-makers/agents, e.g., power systems \cite{powerSystem}, transportation systems \cite{potentialgame},\cite{transportationSystem}, and human-robot interaction systems \cite{humanRobotInteraction}. The interactions among agents need to be modeled.


To characterize agents' interactions in multi-agent systems (MASs), Markov games have been recently explored \cite{markovGameinComputationOffloading},\cite{markovGame2},\cite{markovGame4},\cite{markovGame5}. One desired outcome in an MG is the Nash equilibrium, which represents a stable status such that no agent has the incentive to unilaterally change their policy \cite{defNE}. To solve an MG, multi-agent reinforcement learning (MARL) is needed. Many existing MARL algorithms have been successful in reaching a stationary point, which is a necessary condition for an NE. Specifically, a Nash deep Q-network is developed in \cite{nashDeepQ-network} to handle the complexity and coordination challenges of large-scale traffic signal control. An actor-critic structure is used in \cite{existMARL2} to approximate the Q-function based on each agent's local information to address the curse of dimensionality in MASs.  The optimization and convergence properties of gradient-based algorithms for MGs are rigorously analyzed in \cite{gradienPlayLiNa} to provide theoretical foundations of MARL. It is shown in \cite{gradienPlayLiNa} that the convergence to a NE of a MARL algorithm is typically very challenging for a general-sum MG. 



One possible approach to address the NE-seeking challenge is to formulate the MG as a Markov potential game (MPG). An MPG extends the static potential games to a dynamic setting with state transitions. In a static potential game, a unilateral deviated action by an agent leads to the same amount of change in the potential function and in  the agent's reward function\cite{potentialgame}. Likewise, in an MPG, there exists a potential function that tracks the change of each agent's cumulative rewards. An MPG has appealing properties such as the guaranteed existence of at least one pure-strategy NE and the assured convergence to an NE under gradient play \cite{gradienPlayLiNa}. 
However, an open question remains: given an MAS, how to construct an MPG\cite{gradienPlayLiNa}.


In this paper, we develop sufficient conditions under which an MG is an MPG. 
The contributions of this paper include:
\begin{enumerate}
    \item We provide sufficient conditions on the reward design and on the MDP such that an MG is an MPG.
    \item We apply the MDP and MARL framework to autonomous driving applications. Statistical studies are conducted to evaluate the performance. 
    \item Comparative results between single-agent RL and MARL are provided, highlighting better robustness performance of the MPG-based MARL. 
\end{enumerate}

The remainder of this paper is organized as follows. Section \ref{Section2} introduces preliminaries on MGs, MARL, and relevant solution concepts. Section \ref{Section3} defines MPGs and provides the MPG construction approach. 
Section \ref{Section4} reports the numerical results using autonomous driving as an example, and Section \ref{Section5} concludes the paper.

\section{MARKOV GAME AND MULTI-AGENT REINFORCEMENT LEARNING}\label{Section2}
We define Markov games in Section \ref{sec:MG} and multi-agent RL in Section \ref{sec:MARL}.

\subsection{Markov Game}\label{sec:MG}

A Markov game is defined as a tuple $\mathcal{M} = (\mathcal{N}, \mathcal{S}, \mathcal{A} , P, r, \gamma, \rho)$, where $\mathcal{N} = \{1, 2, \cdots, N\}$ is the set of agents; $\mathcal{S}=\mathcal{S}_1\times\cdots\times\mathcal{S}_N$ is a finite set of states and $\mathcal{A} = \mathcal{A}_1\times\cdots\times\mathcal{A}_N$ is a finite set of actions, where $\mathcal{S}_i$ and $\mathcal{A}_i$ represents the state and action space for each agent $i \in \mathcal{N}$, respectively. The transition model is represented by $P$, where $P(s'|s, a)$ is the probability of transitioning into state $s'$ from $s$ when  $a= (a_1,\cdots, a_N)$ is taken. The reward function $r = (r_1, \cdots, r_N)$ assigns a reward $r_i:\mathcal{S}\times\mathcal{A}\rightarrow\mathbb{R}$ to each agent $i$. The discount factor $\gamma \in [0,1)$ weighs future versus immediate rewards, and $\rho$ is the distribution of the initial state.

Agents select actions based on a policy function $\pi:\mathcal{S}\rightarrow\Delta(\mathcal{A})$, where $\Delta(\mathcal{A})$ is the probability simplex. Consider a decentralized policy $\pi= \pi_1 \times \cdots \times \pi_N$, where each agent takes its own action independently, regardless of other agents' decisions. In other words, at a time step $t$, given the global state $s_t=(s_{1,t},\cdots,s_{N,t})$ and joint actions $a_t=(a_{1,t},\cdots,a_{N,t})$, one has:
\begin{equation}\label{eq:gloabalPolicy}
\pi(a_t|s_t)=\prod_{i=1}^{N}\pi_i(a_{i,t}|s_t).
\end{equation}

We consider a direct parameterization to each agent's policy with $\theta_i$:
\begin{equation}\label{eq:parametrization}
    \pi_{i,\theta_i}(a_i|s)=\theta_{i,(s,a_i)},\quad i=1,2,\cdots,N.
\end{equation}
With a slight abuse of notation, we may use $\theta_i$ and $\theta$ to refer to the parameterized policy $ \pi_{i,\theta_i}$ and $\pi_{\theta}$ respectively when no confusion.
Here $\theta_i\in\Delta(\mathcal{A}_i)^{|\mathcal{S}|}$ with $|\mathcal{S}|$ being the cardinality of $\mathcal{S}$. We denote the feasible set of $\theta_i$ and $\theta$ as $\mathcal{X}_i=\Delta(\mathcal{A}_i)^{|\mathcal{S}|}$ and $\mathcal{X}=\mathcal{X}_1\times\cdots\times\mathcal{X}_N$, respectively.

We assume that the agents can observe the overall state and action information. We denote agent $i$'s trajectory as $\tau=(s_t,a_t,r_{i,t})_{t=0}^{\infty}$, where $a_t\sim\pi_{\theta}(\cdot|s_t),s_{t+1}\sim P(\cdot|s_t,a_t)$.  The value function of agent $i$, $V_i^{\theta}:\mathcal{S}\rightarrow\mathbb{R}$, is defined as the discounted sum of future rewards from the initial state, i.e.,
\begin{equation}\label{eq:valueFunction}
    V_i^{\theta}(s)\coloneqq\mathbb{E}\Biggl[\sum_{t=0}^{\infty}\gamma^{t}r_i(s_t,a_t)\Bigg|\pi_{\theta},s_0=s\Biggr].
\end{equation}
We define agent $i$'s total rewards $J_i:\mathcal{X}\rightarrow\mathbb{R}$ as
\begin{equation}\label{eq:objectiveFunction}
    J_i(\theta)=J_i(\theta_i,\theta_{-i})=J_i(\theta_1,\cdots,\theta_N)\coloneqq\mathbb{E}_{s_0\sim\rho}V_i^{\theta}(s_0),
\end{equation}
where $-i$ represents the set of all other agents except for agent $i$.
For agent $i$, we use $\nabla_{\theta_i}J_i(\theta_i,\theta_{-i})$ to represent the gradient of the total rewards with respect to its policy. 

A Nash Equilibrium solution is a policy profile where no agent has the incentive to deviate from their policy.
\begin{defn}\label{df:nashEquilibrium}(Nash equilibrium \cite{fudenberg1991game})
    A policy $\theta^*=(\theta_1^*,\cdots,\theta_N^*)$ is called a Nash equilibrium if
    \begin{equation}
    J_i(\theta_i^*,\theta_{-i}^*)\geq{J_i(\theta_i',\theta_{-i}^*)},\quad\forall \theta_i' \in \mathcal{X}_i,\quad i \in \mathcal{N}.
    \end{equation}

The NE is called a strict NE if the inequality is strictly satisfied for any deviated policy $\theta_i'\neq\theta_i^{*}\in\mathcal{X}_i$ and any agent $i\in\mathcal{N}$. If the NE $\theta^{*}$ is deterministic, it is a pure NE; otherwise, it is a mixed NE.
\end{defn}

We define the discounted visitation measure  $d_{\theta}$ under a given policy over the states $s$ as \cite{agarwal2019reinforcement}: 
\begin{equation}\label{eq:visitationMeasure}
    d_{\theta}(s)\coloneqq\mathbb{E}_{s_0\sim\rho}(1-\gamma)\sum_{t=0}^{\infty}\gamma^{t}\text{Pr}^{\theta}(s_t=s|s_0),
\end{equation}\\
where $\text{Pr}^{\theta}(s_t=s|s_0)$ indicates the probability of state $s$ being visited when the agents are initialized by $s_0$ and are with the policy $\pi_{\theta}$.

We make the following assumption throughout the paper.
\begin{assump}\label{as:positiveProbability}
    The MG $\mathcal{M}$ satisfies: $d_{\theta}(s)>0, \forall s \in \mathcal{S}, \forall \theta \in \mathcal{X}$.
\end{assump}
This assumption requires that each state in the state space is visited at least once, which is commonly used in RL convergence analysis \cite{globalConvergence}.

\subsection{Multi-Agent Reinforcement Learning}\label{sec:MARL}

By applying direct distributed parameterization to the decentralized structure, we can obtain the gradient ascent algorithm for each player:
\begin{equation}\label{eq:gradientPlay}
    \theta_i^{(t+1)}=\text{Proj}_{\mathcal{X}_i}(\theta_i^{(t)}+\eta\nabla_{\theta_i}J_i(\theta^{(t)})),\enspace\eta>0,
\end{equation}
where $\eta$ represents the learning rate.

\begin{defn}\label{df:firstOrderStationaryPolicy}
(First-order stationary policy\cite{gradienPlayLiNa}) A policy $\theta^*=(\theta_1^*, \cdots, \theta_N^{*})$ is called a first-order stationary policy if $(\theta_i'-\theta_i^{*})^{\top}\nabla_{\theta_i}J_i(\theta^{*})\leq0,\enspace\forall\theta_i'\in\mathcal{X}_i,\enspace i\in \mathcal{N}$.
\end{defn}

Next, we introduce the gradient domination property, which shall play an important role in showing the equivalence between the NE and first-order stationary policy.
\begin{lemma}\label{le:gradientDomination}
    (Gradient domination\cite{gradienPlayLiNa}) For direct distributed parameterization \eqref{eq:parametrization}, the following inequality holds for any $\theta=(\theta_1,\cdots,\theta_N)\in \mathcal{X}$ and any $\theta_i'\in\mathcal{X}_i,i\in\mathcal{N}$:
    \begin{equation}\label{eq:gradientDomination}
        J_i(\theta_i',\theta_{-i})-J_i(\theta_i,\theta_{-i})\leq\left\|\frac{d_{\theta'}}{d_\theta}\right\|_{\infty}\max_{\overline{\theta}_i\in\mathcal{X}_i}(\overline{\theta}_i-\theta_i)^{\top}\nabla_{\theta_i}J_i(\theta),
    \end{equation}
    where$\|\frac{d_{\theta'}}{d_\theta}\|_{\infty}\coloneqq\max_s\frac{d_{\theta'}(s)}{d_{\theta}(s)}$, and $\theta'=(\theta_i',\theta_{-i})$.
\end{lemma}

The inequality \eqref{eq:gradientDomination}  holds when $\theta_{-i}$ is fixed, therefore it leads to the following equivalence between the NE and first-order stationary policy.
\begin{thm}\label{thm:equivalence}
    (Theorem 1\cite{gradienPlayLiNa}) Under Assumption \ref{as:positiveProbability}, first-order stationary policies and NEs are equivalent.
\end{thm}

The proof of Theorem \ref{thm:equivalence} follows from \cite{gradienPlayLiNa}. For completeness, we provide a brief sketch of the proof.

\begin{proof}
    First, we prove all Nash equilibria are first-order stationary policies. According to the Nash equilibrium definition (i.e., Definition \ref{df:nashEquilibrium}),  for any $\theta_i\in\mathcal{X}_i$:
    \begin{equation}\label{preq:thmEquivalence1}
    \begin{split}
        &J_i((1-\delta)\theta_i^{*}+\delta\theta_i,\theta_{-i}^{*})-J_i(\theta_i^{*},\theta_{-i}^{*})\\
        &=\delta(\theta_i-\theta_{-i}^{*})^{\top}\nabla_{\theta_i}J_i(\theta^{*})+o(\delta\left\|\theta_i-\theta_i^{*}\right\|)\leq0,\quad\forall\delta>0.
    \end{split}
    \end{equation}
    As $\delta\rightarrow0$, Eq. \eqref{preq:thmEquivalence1} gives the first-order stationarity condition:
    \begin{equation}\label{preq:thmEquivalence2}
        (\theta_i-\theta_i^{*})^{\top}\nabla_{\theta_i}J_i(\theta^{*})\leq0,\quad\forall\theta_i\in\mathcal{X}_i.
    \end{equation}
    
    Now we show that  first order stationary policies are Nash equilibria. From Assumption \ref{as:positiveProbability} we know that for any pair of parameters $\theta'=(\theta_i',\theta_{-i}^{*})$  and $\theta^{*}=(\theta_i^{*},\theta_{-i}^{*})$, we have $\left\|\frac{d_{\theta'}}{d_{\theta^{*}}}\right\|<+\infty$. From Lemma \ref{le:gradientDomination}, we have that for any first-order stationary policy $\theta^{*}$, \begin{equation}\label{preq:thmEquivalence3}
    \begin{split}
        &J_i(\theta_i',\theta_{-i}^{*})-J_i(\theta_i^{*},\theta_{-i}^{*})\\&\leq\left\|\frac{d_{\theta'}}{d_{\theta^{*}}}\right\|_{\infty}\max_{\overline{\theta}_i\in\mathcal{X}_i}(\overline{\theta}_i-\theta_i^{*})^{\top}\nabla_{\theta_i}J_i(\theta^{*})\leq0,
    \end{split}
    \end{equation}
    which completes the proof.
\end{proof}

Given the equivalence of NEs and stationary points, the remaining question is whether the MG can converge to a stationary policy under gradient play. One major reason for a failure is that the vector field $\{\nabla_{\theta_i}J_i(\theta)\}_{i=1}^{N}$ is not a conservative field, which can be addressed by making the MG an MPG \cite{gradienPlayLiNa}.

\section{MARKOV POTENTIAL GAME}\label{Section3}
We define an MPG and include its properties in Section \ref{sec:MPG&properties} and provide sufficient conditions for MPG construction in Section \ref{sec:sufficientConditions}.
\subsection{Definition and Properties of MPGs}\label{sec:MPG&properties}

\begin{defn}\label{df:markovPotentialGame}
    (Markov potential game \cite{globalConvergence}) An MG $\mathcal{M}$ is called an MPG if there exists a potential function $\phi:\mathcal{S}\times\mathcal{A}\rightarrow\mathbb{R}$ such that for any agent $i$ and any pair of policy parameters $(\theta_i',\theta_{-i}), (\theta_i,\theta_{-i})$ at any state $s$:
    \begin{equation}\label{eq:markovPotentialGame}
    \begin{split}
        &\mathbb{E}\Biggl[\sum_{t=0}^{\infty}\gamma^{t}r_i(s_t,a_t)\Bigg|\pi_{(\theta_i',\theta_{-i})},s_0=s\Biggr]\\
        &-\mathbb{E}\Biggl[\sum_{t=0}^{\infty}\gamma^{t}r_i(s_t,a_t)\Bigg|\pi_{(\theta_i,\theta_{-i})},s_0=s\Biggr]\\
        =&\mathbb{E}\Biggl[\sum_{t=0}^{\infty}\gamma^{t}\phi(s_t,a_t)\Bigg|\pi_{(\theta_i',\theta_{-i})},s_0=s\Biggr]\\
        &-\mathbb{E}\Biggl[\sum_{t=0}^{\infty}\gamma^{t}\phi(s_t,a_t)\Bigg|\pi_{(\theta_i,\theta_{-i})},s_0=s\Biggr].
    \end{split}
    \end{equation}
\end{defn}

The total potential function of the MPG can then be defined as:
\begin{equation}\label{eq:totalPotentialFunction}
\Phi(\theta)\coloneqq\mathbb{E}_{s_0\sim\rho}\left[\sum_{t=0}^{\infty}\gamma^t\phi(s_t,a_t)\Bigg|\pi_{(\theta_i,\theta_{-i})},s_0=s\right].
\end{equation}

Proposition \ref{pp:aloGlobalMaximum} ensures the existence of at least one pure NE in an MPG.
\begin{prop}\label{pp:aloGlobalMaximum}
    (Proposition 1\cite{gradienPlayLiNa}) For an MPG, there is at least one global maximum $\theta^{*}$ of the total potential function $\Phi$, i.e., $\theta^{*}\in\argmax_{\theta\in\mathcal{X}}\Phi(\theta)$, that is a pure NE.
\end{prop}

By combining Eq. \eqref{eq:valueFunction}, \eqref{eq:objectiveFunction} and \eqref{eq:totalPotentialFunction}, we can rewrite \eqref{eq:markovPotentialGame} as:
\begin{equation}\label{eq:rewrittenMPGDefinition}
    J_i(\theta_i',\theta_{-i})-J_i(\theta_i,\theta_{-i})=\Phi(\theta_i',\theta_{-i})-\Phi(\theta_i,\theta_{-i}).
\end{equation}

Further,
\begin{equation}
    \nabla_{\theta_i}J_i(\theta)=\nabla_{\theta_i}\Phi(\theta),
\end{equation}
by which we can observe that instead of gradient play, we can run the following projected gradient ascent with respect to the total potential function $\Phi$:
\begin{equation}\label{eq:gradientPlayTotal}
    \theta^{(t+1)}=\text{Proj}_{\mathcal{X}}(\theta^{(t)}+\eta\nabla_{\theta}\Phi(\theta^{(t)})),\enspace\eta>0,
\end{equation}

Theorem \ref{thm:globalConvergence} ensures the convergence to an NE under gradient play in an MPG.
\begin{thm}\label{thm:globalConvergence}
    (\cite[Theorem 4.2]{globalConvergence}) Given an MPG, for any initial state, the projected gradient ascent in Eq. \eqref{eq:gradientPlayTotal} converges to an NE as $t\rightarrow\infty$.
\end{thm}
\subsection{Construction of MPG}\label{sec:sufficientConditions}

In this subsection, we develop sufficient conditions to construct the MPG. Note that Eq. \eqref{eq:markovPotentialGame} suggests that the difference in the discounted sum of future rewards caused by agent $i$'s deviated policy is the same as the difference in the discounted sum of future values of the potential function.

Theorem \ref{thm:selfPotentialFunction} indicates that when agent $i$'s transition probability and reward function are respectively only determined by its own policy, an MG is an MPG.
\begin{thm}\label{thm:selfPotentialFunction}
    Consider an MG where each agent has independent initial state distribution, and agent $i$'s reward function satisfies the following form,
    \begin{equation}\label{eq:selfReward}
        r_i(s_t,a_t)=r_i^{self}(s_{i,t},a_{i,t}),
    \end{equation}
    where $r_i^{self}(s_{i,t},a_{i,t})$ is solely dependent on agent $i$'s policy $\theta_i$. Suppose $P(s_i'|s_i,a_i,a_{-i}')=P(s_i'|s_i,a_i,a_{-i})$, $\forall a_{-i},a_{-i}'\in\mathcal{A}_{-i},\forall a_i\in\mathcal{A}_i,\forall s_i\in\mathcal{S}_i$ and $\forall i\in\mathcal{N}$. Then the formulated game is an MPG with a potential function \begin{equation}\label{eq:selfPotenitalFunction}
        \phi^{self}(s_t,a_t)=\sum_{i\in\mathcal{N}}r_i^{self}(s_{i,t},a_{i,t}).
    \end{equation}
\end{thm}

\begin{proof}
    With \eqref{eq:selfReward}, the total rewards of agent $i$ is:
    \begin{equation}\label{preq:selfTotalReward}
    \begin{split}
        J_i(\theta)&=\mathbb{E}_{s_0\sim\rho}\Biggl[\sum_{t=0}^{\infty}\gamma^{t}r_i^{self}(s_{i,t},a_{i,t})\Bigg|\pi_{\theta},s_0\Biggr].
    \end{split}
    \end{equation}
    
Therefore,
    \begin{equation}
    \begin{split}
       &J_i(\theta_i',\theta_{-i})-J_i(\theta_i,\theta_{-i})\\
        =&\mathbb{E}_{s_0\sim\rho}\Biggl\{\sum_{t=0}^{\infty}\gamma^{t}\Biggl[r_i^{self}(s_{i,t}',a_{i,t}')\\
        &-r_i^{self}(s_{i,t},a_{i,t})\Bigg|\pi_{(\theta_i',\theta_{-i})},\pi_{(\theta_i,\theta_{-i})},s_0\Biggr]\Biggr\}.
    \end{split}
    \end{equation}

    Meanwhile, the total potential function is:
    \begin{equation}\label{preq:selfTotalPotentialFunction}
    \Phi(\theta)=\mathbb{E}_{s_o\sim\rho}\Biggl[\sum_{t=0}^{\infty}\gamma^{t}\sum_{i\in\mathcal{N}}r_i^{self}(s_{i,t},a_{i,t})\Bigg|\pi_{\theta},s_0\Biggr].
    \end{equation}
    
    As is defined in \eqref{preq:selfTotalPotentialFunction}, the deviation in the policy of agent $i$ yields

    \begin{equation}\label{preq:rightSelfPotentialFunction}
    \begin{split}
    &\Phi(\theta_i',\theta_{-i})-\Phi(\theta_i,\theta_{-i})\\
    &=\, \mathbb{E}_{s_0\sim\rho}\Biggl\{
    \sum_{t=0}^{\infty}\gamma^{t}\Biggl[
    r_i^{self}(s_{i,t}',a_{i,t}') \\
    &+\sum_{\substack{j\in\mathcal{N}\\ j\neq i}}r_j^{self}(s_{j,t},a_{j,t})
    - r_i^{self}(s_{i,t},a_{i,t}) \\
    &-\sum_{\substack{j\in\mathcal{N}\\ j\neq i}}r_j^{self}(s_{j,t},a_{j,t})
    \Bigg|\pi_{(\theta_i',\theta_{-i})},\pi_{(\theta_i,\theta_{-i})},s_0\Biggr]
    \Biggr\} \\
    &=\, \mathbb{E}_{s_0\sim\rho}\Biggl\{
    \sum_{t=0}^{\infty}\gamma^{t}\Biggl[
    r_i^{self}(s_{i,t}',a_{i,t}')\\
    &-r_i^{self}(s_{i,t},a_{i,t})
    \Bigg|\pi_{(\theta_i',\theta_{-i})},\pi_{(\theta_i,\theta_{-i})},s_0\Biggr]
    \Biggr\}.
    \end{split}
    \end{equation}

    The term $\sum_{j\in\mathcal{N},j\neq i}r_j^{self}(s_{j,t},a_{j,t})$ can be separated out and canceled out  because  $P(s_j'|s_j,a_j,a_{-j}')=P(s_j'|s_j,a_j,a_{-j})$, i.e., as $a_{i,t}$ changes to $a_{i,t}'$, the trajectory of agent $j$ will not be affected.
    
   Therefore, Eq. \eqref{eq:rewrittenMPGDefinition} holds, and the proof is complete. 
\end{proof}


Theorem \ref{thm:jointPotentialFunction} considers scenarios where agent $i$'s reward depends on both its own policy and other agents' policies.

\begin{thm}\label{thm:jointPotentialFunction}
    Consider an MG where each agent has independent initial state distribution, and agent $i$'s reward function satisfies the following form,
    \begin{equation}\label{eq:jointReward}
        r_i(s_t,a_t)=\sum_{j\in\mathcal{N},j\neq i}r_{ij}(s_{i,t},s_{j,t},a_{i,t},a_{j,t}),
    \end{equation}
    where $r_{ij}(s_{i,t},s_{j,t},a_{i,t},a_{j,t})=r_{ji}(s_{j,t},s_{i,t},a_{j,t},a_{i,t})$, $\forall i,$ $j\in \mathcal{N}, i\neq j$. Suppose $P(s_i'|s_i,a_i,a_{-i}')=P(s_i'|s_i,a_i,a_{-i})$, $\forall a_{-i},a_{-i}'\in\mathcal{A}_{-i},\forall a_i\in\mathcal{A}_i,\forall s_i\in\mathcal{S}_i$ and $\forall i\in\mathcal{N}$. Then the formulated game is an MPG with a potential function  \begin{equation}\label{eq:jointPotenitalFunction}
        \phi^{joint}(s_t,a_t)=\sum_{i\in\mathcal{N}}\sum_{j\in\mathcal{N},j<i}r_{ij}(s_{i,t},s_{j,t},a_{i,t},a_{j,t}).
    \end{equation}
\end{thm}
\begin{proof}
    With \eqref{eq:jointReward}, the total rewards of agent $i$ is:
    \begin{equation}\label{preq:jointTotalReward}
    \begin{split}
          &J_i(\theta)\\
          &=\mathbb{E}_{s_0\sim\rho}\Biggl[\sum_{t=0}^{\infty}\gamma^{t}\sum_{j\in\mathcal{N},j\neq i}r_{ij}(s_{i,t},s_{j,t},a_{i,t},a_{j,t})\Bigg|\pi_{\theta},s_0\Biggr].
    \end{split}
    \end{equation}
    
Therefore,
    \begin{equation}
    \begin{split}
        &J_i(\theta_i',\theta_{-i}) - J_i(\theta_i,\theta_{-i})\\
        &= \mathbb{E}_{s_0\sim\rho}\Biggl\{ \sum_{t=0}^{\infty}\gamma^{t} \Biggl[
        \sum_{j\in\mathcal{N},\, j\neq i} \Bigl(
        r_{ij}(s_{i,t}',s_{j,t},a_{i,t}',a_{j,t})\\
        &- r_{ij}(s_{i,t},s_{j,t},a_{i,t},a_{j,t})
        \Bigr)\Bigg|\pi_{(\theta_i',\theta_{-i})},\pi_{(\theta_i,\theta_{-i})},s_0\Biggr]\Biggr\}.
    \end{split}
    \end{equation}

    Meanwhile, the total potential function is:
    \begin{equation}\label{preq:jointPotentialFunction}
    \begin{split}
    &\Phi(\theta)=\mathbb{E}_{s_o\sim\rho}\\
    &\Biggl[\sum_{t=0}^{\infty}\gamma^{t}\sum_{i\in\mathcal{N}}\sum_{j\in\mathcal{N},j<i}r_{ij}(s_{i,t},s_{j,t},a_{i,t},a_{j,t})\Bigg|\pi_{\theta},s_0\Biggr].
    \end{split}
    \end{equation}
    
    As is defined in \eqref{eq:jointPotenitalFunction}, the deviation in the policy of agent $i$ yields:

    \begin{equation}\label{preq:rightJointPotentialFunction}
    \begin{split}
        &\Phi(\theta_i',\theta_{-i})-\Phi(\theta_i,\theta_{-i})\\
        &= \mathbb{E}_{s_0\sim\rho}\Biggl\{
        \sum_{t=0}^{\infty}\gamma^{t}\, \sum_{j\in\mathcal{N},\, j\neq i}\Biggl[
        r_{ij}(s_{i,t}',s_{j,t},a_{i,t}',a_{j,t})\\
        &-\, r_{ij}(s_{i,t},s_{j,t},a_{i,t},a_{j,t})
        \Bigg|\pi_{(\theta_i',\theta_{-i})},\pi_{(\theta_i,\theta_{-i})},s_0\Biggr]
        \Biggr\}.
    \end{split}
    \end{equation}

    Recall the symmetry of the joint reward function, i.e., $r_{ij}(s_{i,t},s_{j,t},a_{i,t},a_{j,t})=r_{ji}(s_{j,t},s_{i,t},a_{j,t},a_{i,t})$ and the condition $P(s_j'|s_j,a_j,a_{-j}')=P(s_j'|s_j,a_j,a_{-j})$, $\forall i, j\in \mathcal{N}, i\neq j$ , i.e., as $a_{i,t}$ changes to $a_{i,t}'$, the trajectory of agent $j$ will not be affected. As such, the terms that do not concern agent $i$ will be canceled, and thus Eq. \eqref{preq:rightJointPotentialFunction} holds.

\end{proof}

Theorem \ref{thm:selfAndJointPotentialFunction} combines the conditions in Theorems \ref{thm:selfPotentialFunction} and \ref{thm:jointPotentialFunction}.

\begin{thm}\label{thm:selfAndJointPotentialFunction}
    Consider an MG where each agent has independent initial state distribution, and agent $i$'s reward function satisfies the following form,
    \begin{equation}\label{eq:mixReward}
    \begin{split}
        &r_i(s_t,a_t)=\alpha r_i^{self}(s_{i,t},a_{i,t})\\
        &+\beta\sum_{j\in\mathcal{N},j\neq i}r_{ij}(s_{i,t},a_{i,t},s_{j,t},a_{j,t}),
    \end{split}
    \end{equation}
    where $r_i^{self}(s_{i,t},a_{i,t})$ and $\sum_{j\in\mathcal{N},j\neq i}r_{ij}(s_{i,t},a_{i,t},s_{j,t},a_{j,t})$ follow directly from \eqref{eq:selfReward} and \eqref{eq:jointReward}, respectively, and $\alpha\in\mathbb{R}$ and $\beta\in\mathbb{R}$. Suppose $P(s_i'|s_i,a_i,a_{-i}')=P(s_i'|s_i,a_i,a_{-i})$, $\forall a_{-i},a_{-i}'\in\mathcal{A}_{-i},\forall a_i\in\mathcal{A}_i,\forall s_i\in\mathcal{S}_i$ and $\forall i\in\mathcal{N}$. Then the formulated game is an MPG with a potential function
    \begin{equation}\label{eq:selfAndJointPotentialFunction}
    \begin{split}
        &\phi(s_t,a_t)=\alpha\phi^{self}(s_t,a_t)+\beta\phi^{joint}(s_t,a_t)\\
        &=\alpha \sum_{i\in\mathcal{N}}r_i^{self}(s_{i,t},a_{i,t})\\
        &+\beta\sum_{i\in\mathcal{N}}\sum_{j\in\mathcal{N},j< i}r_{ij}(s_{i,t},s_{j,t},a_{i,t},a_{j,t})).
    \end{split}
    \end{equation}
\end{thm}
\begin{proof}
    The proof can be completed by applying Theorem \ref{thm:selfPotentialFunction} and Theorem \ref{thm:jointPotentialFunction}.
\end{proof}

\section{APPLICATION TO AUTONOMOUS DRIVING}\label{Section4}
In this section, we evaluate the performance of the MPG-based MARL using autonomous driving applications with intersection-crossing scenarios.
\subsection{Simulation Setup}
\begin{figure}
    \centering
    \includegraphics[width=0.5\linewidth]{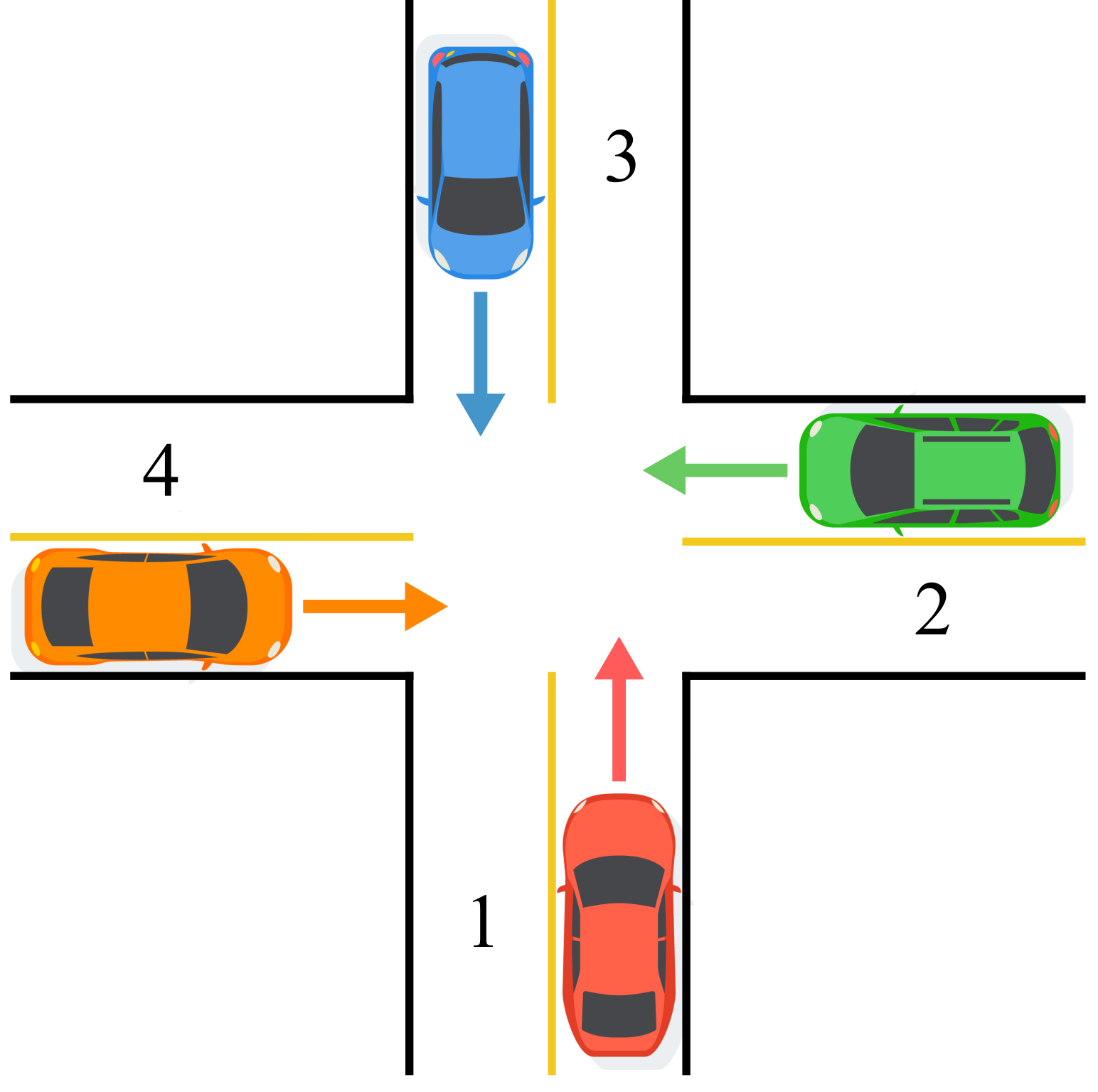}
    \caption{The four-vehicle intersection scenario.}
    \label{fig:intersectionCrossing}
\end{figure}
Consider a four-way intersection depicted in Figure \ref{fig:intersectionCrossing}, where vehicle ``2" is the ego vehicle, and the rest are surrounding vehicles. Each vehicle is set to go straight in its designated lane, hence only longitudinal actions are analyzed. The state for vehicle $i$ is $s_{i,t}=(p_i(t),v_i(t))$, where $i=1,\cdots,N$ is the index for each vehicle; $p_i$ and $v_i$ represents the longitudinal position and velocity along its lane. Based on the direction of motion, $p_i$ can be $x_i$ or $y_i$, where $x_i$ and $y_i$ are the position of the center of mass of vehicle $i$.

We consider the single-point mass dynamics for each agent: 
\begin{equation}\label{eq:intersectionDynamics}
\begin{split}
    &x_i(t+1)=x_i(t)+v_{i,x}(t)\Delta t,\\
    &v_{i,x}(t+1)=v_{i,x}(t)+a_{i,x}(t)\Delta t,\\
    &y_i(t+1)=y_i(t)+v_{i,y}(t)\Delta t,\\
    &v_{i,y}(t+1)=v_{i,y}(t)+a_{i,y}(t)\Delta t,
\end{split}
\end{equation}
where $v_{i,x}$ and $v_{i,y}$ are the velocities of the center of mass of vehicle $i$ along $x$ and $y$ axes, respectively. The action is the longitudinal acceleration of the vehicle. The sign of the acceleration indicates its direction relative to the defined coordinate axis, where a positive value represents that the acceleration is aligned with the positive direction of that axis. We select $\Delta t=0.5$ s here.
 
Each vehicle's action space is $\mathcal{A}_i=[-g,g]$, where $g=9.81$ $\text{m}/\text{s}^2$ denotes the gravitational constant. The action $a_{i,t}$ can take any real value within this range. Each vehicle is controlled by a deterministic policy defined as $\pi:\mathcal{S}\rightarrow\mathcal{A}$. We estimate the total rewards function over a fixed time horizon $T$,
\begin{equation}\label{simeq:totalPotentialFunction}
J_i(\theta)=\mathbb{E}_{s_0\sim\rho}\left[\sum_{t=0}^{T}\gamma^tr_i(s_t,a_t)\Bigg|\pi_{\theta},s_0\right],
\end{equation}
where we select $T$ to correspond to $20$ s.

We first consider a $4$-vehicle scenario, i.e.,  $N=4$. The initial state $s_0\in\mathcal{S}$ is sampled via stratified sampling to ensure representative coverage of the multidimensional space. Specifically, each element of the state vector (e.g., positions and velocities of individual vehicles) is discretized into evenly spaced intervals, from which one sample is randomly drawn. The initial state set is then constructed by taking the Cartesian product of these sampled values across all dimensions. Next, we derive the NE policies for each vehicle by solving an MPG at each visited state. The driving performance for each vehicle consists of two parts: desired velocity tracking and collision avoidance. Specifically, 
\begin{equation}
    J_i(\theta) = \omega_{i,1}J_i^{self}(\theta)+\omega_{i,2}J_i^{joint}(\theta),
\end{equation}
where $\omega_{i,1}$ and $\omega_{i,2}$ are constant coefficients to balance the two parts. The first term, i.e., $J_i^{self}(\theta)$, is to motivate the vehicle to maintain its desired velocity, and it takes the form \eqref{preq:selfTotalReward}, where $r_i^{self}(s_{i,t},a_{i,t})$ is given by
\begin{equation}
    r_i^{self}(s_{i,t},a_{i,t})=-(v_{i,t}-v_{i,d})^2.
\end{equation}
The desired velocities of vehicles $i=1,2,3,4$ are chosen as $v_{1,d}=5$ m/s, $v_{2,d}=-5$ m/s, $v_{3,d}=-5$ m/s, $v_{4,d}=5$ m/s. The second term $J_i^{joint}(\theta)$ is formulated to prevent collisions between vehicles and takes the form \eqref{preq:jointTotalReward}, where $r_{ij}(s_{i,t},s_{j,t},a_{i,t},a_{j,t})$ is given by
\begin{equation}\label{simeq:jointReward}
\begin{split}
    &r_{ij}(s_{i,t},s_{j,t},a_{i,t},a_{j,t})\\
          &=-\frac{1}{\sqrt{(x_i(t)-x_j(t))^2+(y_i(t)-y_j(t))^2}+\epsilon}.
\end{split}
\end{equation}
The parameter $\epsilon$ is introduced to avoid the denominator being zero and is set to be $1\times10^{-5}$.

According to Theorem \ref{thm:selfAndJointPotentialFunction}, with the reward function design described above, the $N$-player MG qualifies as an MPG.

\subsection{Training Setup and Convergence Analysis}
\begin{figure}
    \centering
    \includegraphics[width=0.95\linewidth]{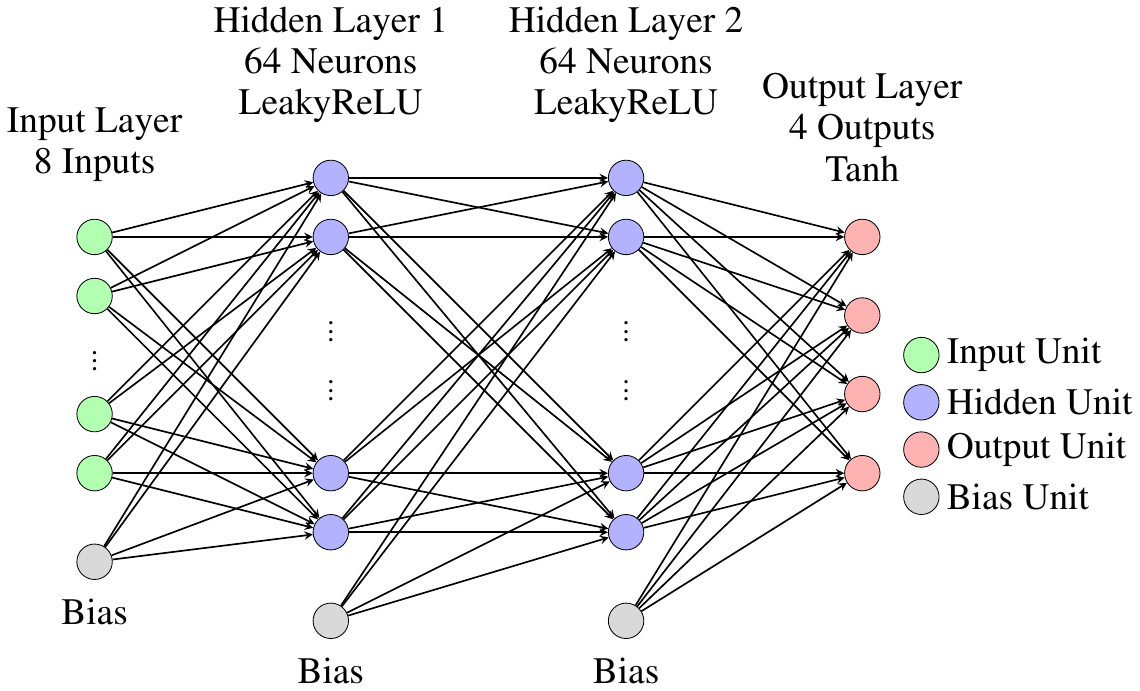}
    \caption{NN architecture.}
    \label{fig:NNArchitecture}
\end{figure}
In our simulation, we use a neural network (NN) to parameterize the deterministic policies of the four vehicles. The NN architecture is shown in Figure \ref{fig:NNArchitecture}, which consists of:
\begin{enumerate}
    \item An input layer with 8 neurons corresponding to the state variables.
    \item Two fully connected hidden layers, each with 64 neurons and LeakyReLU activation functions.
    \item A fully connected output layer producing 4 outputs, followed by a Tanh activation function to constrain the output within $[-1,1]$. The outputs are subsequently scaled by a factor of 9.81 to yield actions within the specified action space.
\end{enumerate}

Training is performed using gradient ascent \eqref{eq:gradientPlayTotal} till the policy reaches convergence. For each episode, the total rewards $J_i(\theta)$ is computed with a discount factor of $\gamma=0.99$.

Gradients with respect to the NN parameters are computed using automatic differentiation, and the parameters are updated using the Adam optimizer with a learning rate of 0.001. This training procedure is implemented using MATLAB’s deep learning tools \cite{matlabDLToolbox} (e.g., dlnetwork, dlfeval, and adamupdate).






\subsection{Evaluation Results in Specific Scenarios}
In this subsection, we evaluate the performance of the derived NE policies in two specific scenarios, with all vehicles following the NE policies.

\begin{figure}[htbp]
    \centering
    
    \begin{subfigure}[b]{0.49\columnwidth}
        \includegraphics[width=\columnwidth]{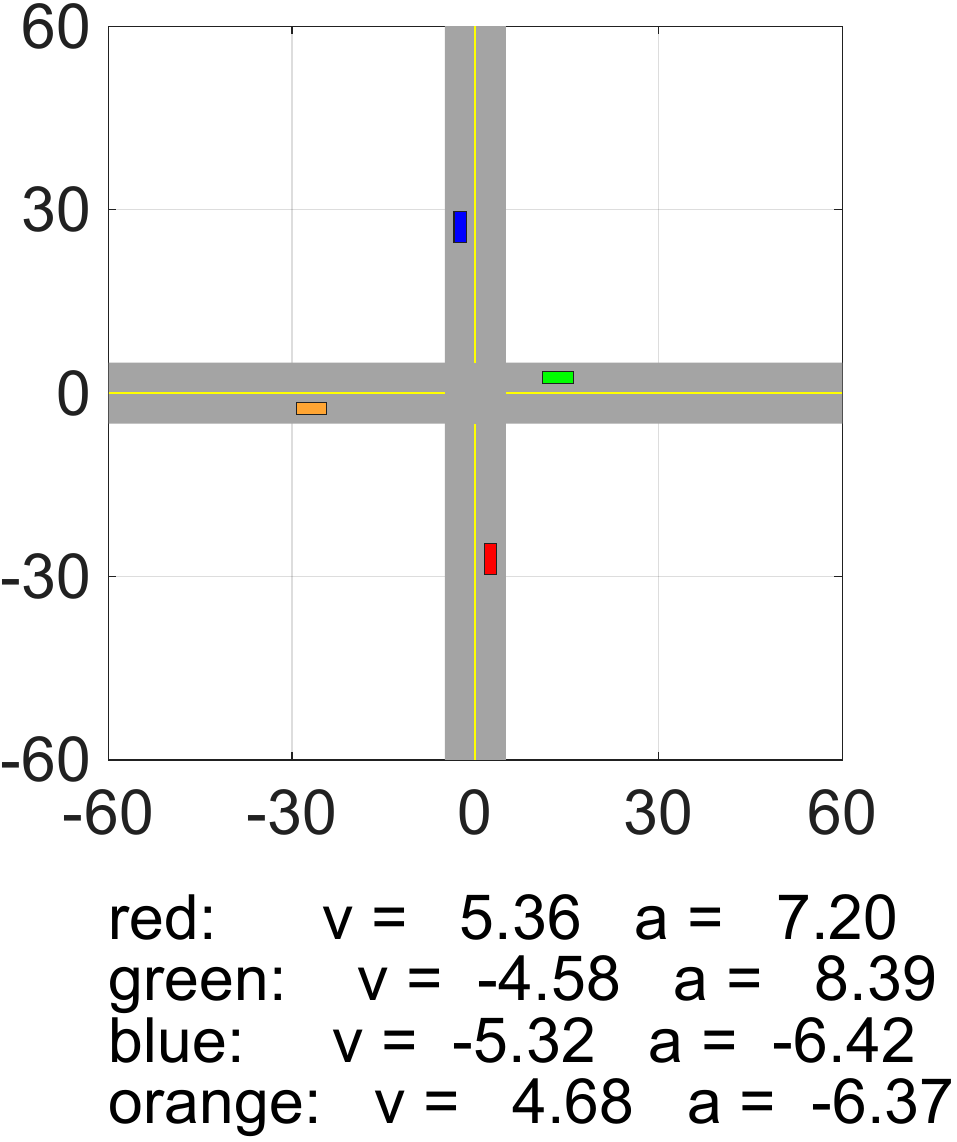}
        \caption{}
        \label{fig:Scenario11}
    \end{subfigure}
    \hfill
    \begin{subfigure}[b]{0.49\columnwidth}
        \includegraphics[width=\columnwidth]{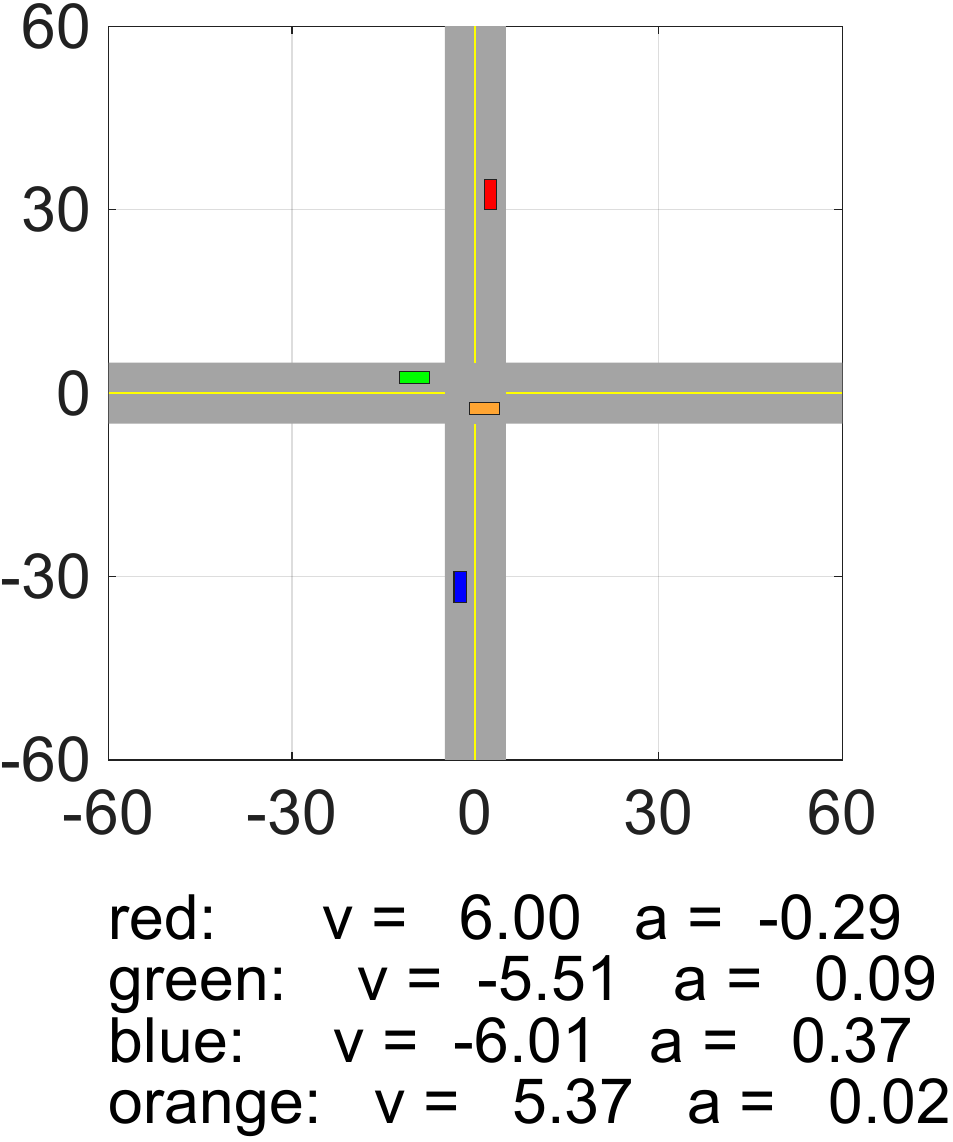}
        \caption{}
        \label{fig:Scenario12}
    \end{subfigure}
    
    \hfill
    
    \begin{subfigure}[b]{0.8\columnwidth}
        \includegraphics[width=\columnwidth]{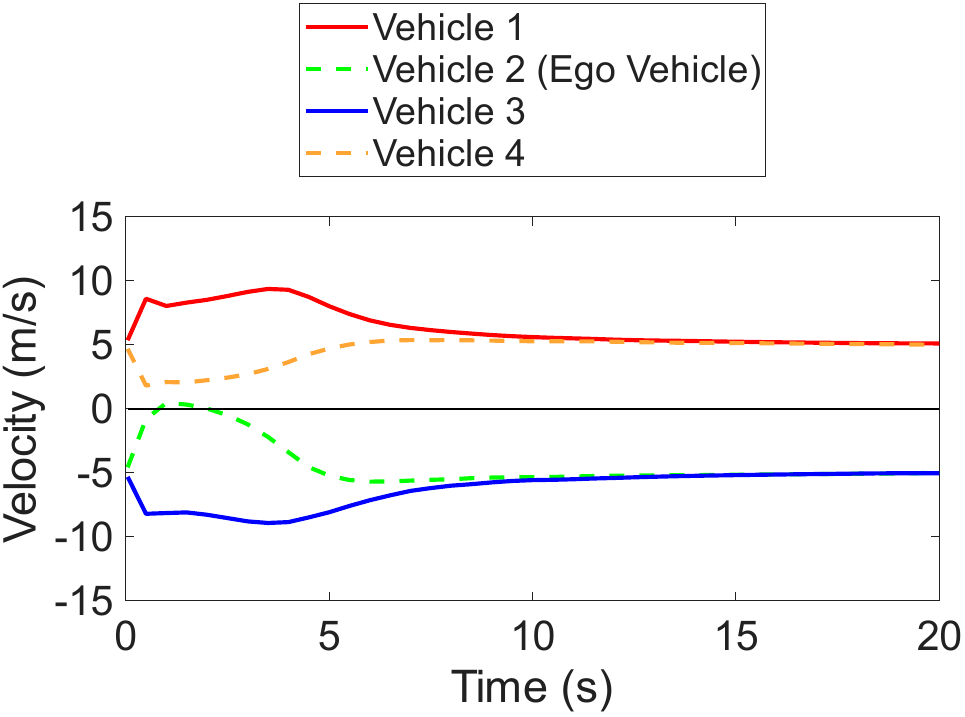}
        \caption{}
        \label{fig:Scenario1Velocity}
    \end{subfigure}
    
    \caption{Vehicles' performance in Scenario 1. (a): The ego vehicle decelerates and waits to avoid collisions; (b): The ego vehicle drives around the desired velocity after crossing; (c):  The velocity histories of all vehicles. }
    \label{fig:Scenario1all}
\end{figure}

\begin{table*}[htbp]
\centering
\caption{Statistical Results: MARL with MPGs}
\label{tab:statisticalResults}
\begin{tabular}{c|c|c|c} 
\hline
Surrounding vehicles' policies & NE & Rule-based policy &  Constant speed \\ \hline
Collision rate & 0/100 & 0/100 & 0/100 \\ \hline
Average ego speed (m/s) & 4.4168 & 3.5252 & 2.5530\\ \hline
\end{tabular}
\end{table*}


\begin{table*}[htbp]
\centering
\caption{Comparative Results: MARL vs. Single-agent RL}
\label{tab:comparativeResults2}
\begin{tabular}{c|c|c|c|c|c|c} 
\hline
 Solution method&\multicolumn{3}{c|}{MPG-based MARL}  &\multicolumn{3}{c}{Single-agent RL} \\ \hline
Surrounding vehicles' strategies & NE & Rule-based policy & Constant speed & NE & Rule-based policy & Constant speed\\ \hline
Collision rate & 0/100 & 0/100 & 0/100 & 11/100 & 0/100 & 45/100\\ \hline
Average ego speed (m/s) & 4.4061 & 3.5179 & 2.5526 & 2.8401 & 2.0595 & 1.7345\\ \hline
\end{tabular}
\end{table*}

\textit{Scenario 1:} In this scenario, we select the initial positions of the vehicles such that the ego vehicle is closer to the center of the intersection compared to the surrounding vehicles. In such a scenario, the ego vehicle executes a substantial deceleration and waits at the intersection to avoid potential collisions, and speeds up afterwards to its desired speed while crossing the intersection. Two key moments are shown in Figures (\ref{fig:Scenario11}) and (\ref{fig:Scenario12}). The velocity histories of all vehicles are shown in Figure (\ref{fig:Scenario1Velocity}). 

\begin{figure}[htbp]
    \centering
    
    \begin{subfigure}[b]{0.49\columnwidth}
        \includegraphics[width=\columnwidth]{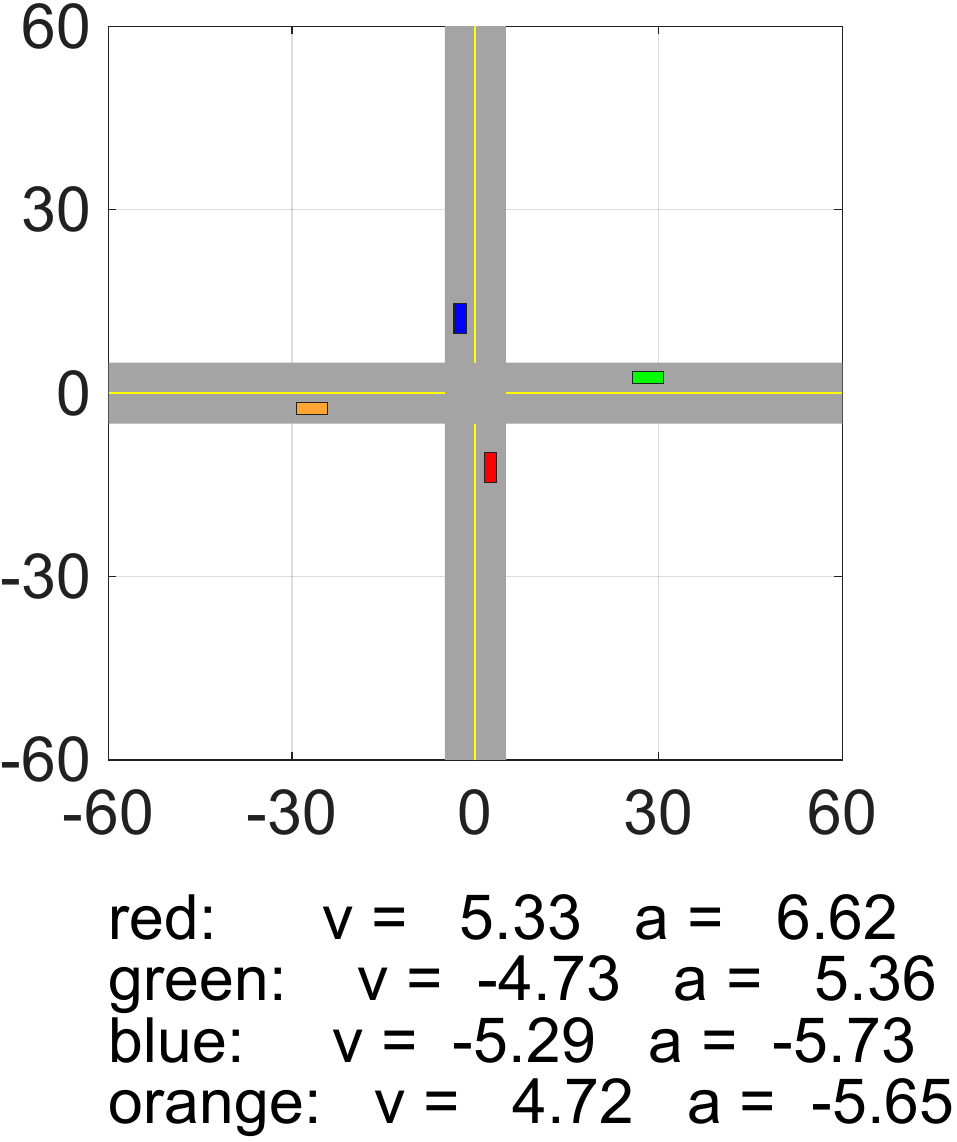}
        \caption{}
        \label{fig:Scenario21}
    \end{subfigure}
    \hfill
    \begin{subfigure}[b]{0.49\columnwidth}
        \includegraphics[width=\columnwidth]{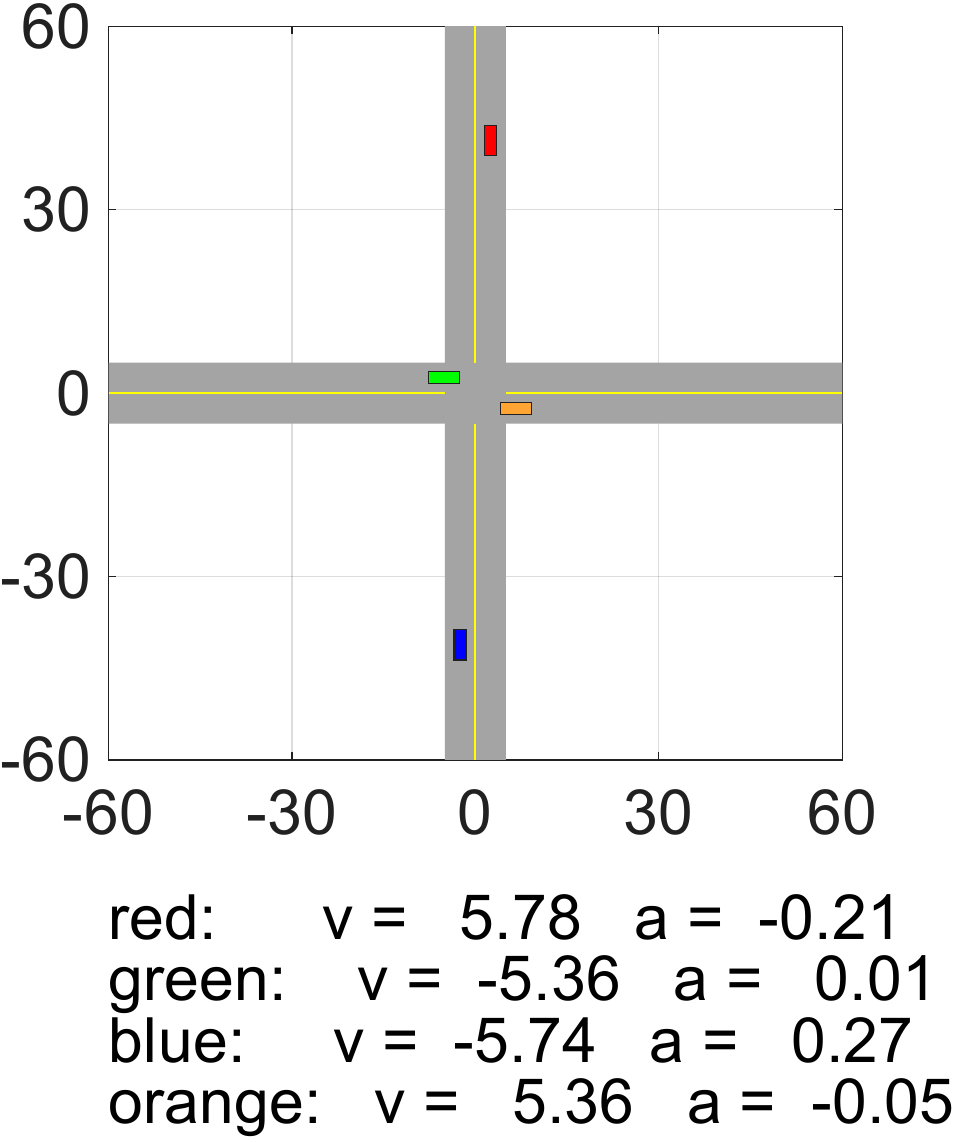}
        \caption{}
        \label{fig:Scenario22}
    \end{subfigure}
    
    \hfill
    
    \begin{subfigure}[b]{0.8\columnwidth}
        \includegraphics[width=\columnwidth]{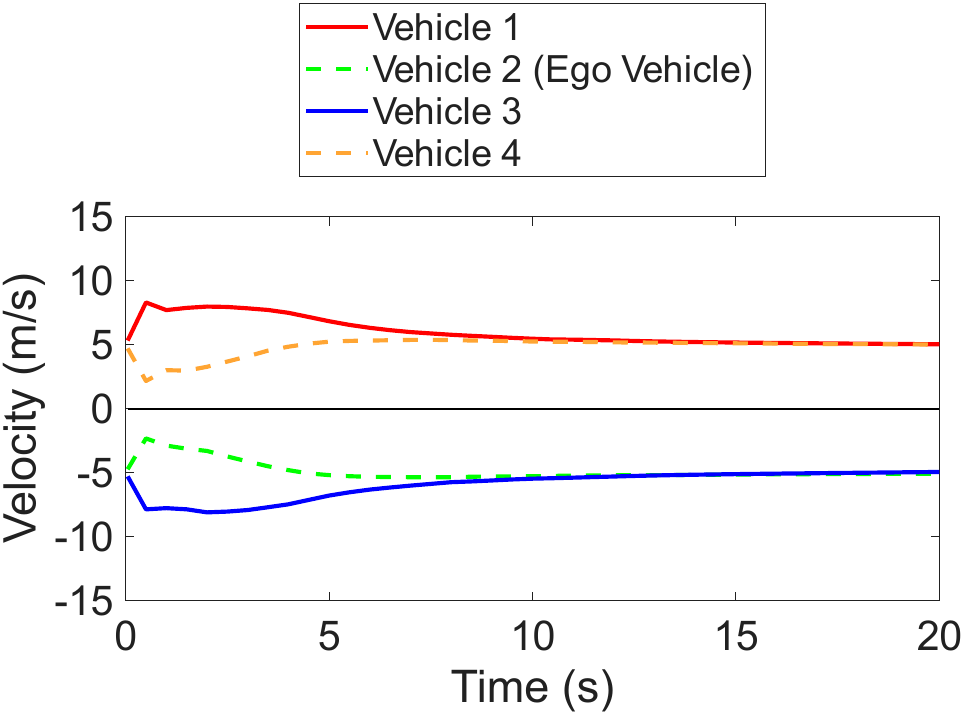}
        \caption{}
        \label{fig:Scenario2Velocity}
    \end{subfigure}
    
    \caption{Vehicles' performance in Scenario 2. (a): The ego vehicle decelerates to yield to the surrounding vehicles; (b): The ego vehicle drives around the desired velocity after crossing; (c):  The velocity histories of all vehicles. }
    \label{fig:Scenario2all}
\end{figure}

\textit{Scenario 2:} In this scenario, we select the initial positions of the vehicles such that vehicle ``1" and ``3", which have trajectory conflicts with the ego vehicle, are closer to the center of the intersection compared to the ego vehicle. In such a scenario, the ego vehicle first yields to vehicle ``1" and ``3" and then speeds up to cross the intersection after the surrounding vehicles have cleared the intersection.  Two key moments are shown in Figures (\ref{fig:Scenario21}) and (\ref{fig:Scenario22}). The velocity histories of all vehicles are shown in Figure (\ref{fig:Scenario2Velocity}).

\subsection{Evaluation Results in Statistical Studies}\label{Section4B}
We conduct statistical studies to comprehensively evaluate the performance of the MARL. 
To evaluate the robustness of the NE, we consider three possible policies for the surrounding vehicles: 1) NE policies, 2) a first-come-first-served rule-based policy, and 3) a constant speed policy. The first policy represents rational and intelligent decision-making. The second policy, while exhibiting some level of rationality, is notably less sophisticated than the first. The third policy is neither intelligent nor safety-conscious, yet it reflects extreme cases where drivers fail to react to dangers or potential collisions due to, e.g., distractions.

We test 100 scenarios with randomized initial states and collect the collision rate and average ego speed. The collision rate is the number of scenarios where a collision with the ego vehicle occurs during the test. The statistical results are shown in Table \ref{tab:statisticalResults}, which leads to the following observations:
\begin{enumerate}
    \item The NE enables the vehicles to cross the intersection safely: No collisions occur out of 100 scenarios when all vehicles use NE policies, when surrounding vehicles use rule-based policy, and even when the surrounding vehicles are safety-agnostic
    , demonstrating satisfying collision avoidance performance. 
    \item The NE enables the vehicles to efficiently cross the intersection, i.e., the vehicle's average speed is the closest to its desired speed, demonstrating satisfying travel efficiency while ensuring safety.
    \item The resulting NE is likely a local NE.
\end{enumerate}

\subsection{Evaluation Results in Comparative Studies}
Next we consider comparative studies on the performance of single-agent RL and MARL. In the single-agent RL, we let the surrounding vehicles take the rule-based policy and train the ego vehicle's optimal policy. For the MARL, we use the potential function optimization algorithm \eqref{eq:gradientPlayTotal}. We then test the two trained policies in three settings, respectively corresponding to the three surrounding vehicle policies. The results are shown in Table \ref{tab:comparativeResults2}. It is observed that compared to the single-agent RL, the MARL has better robustness in terms of lower collision rates when the surrounding vehicles perform unexpected policies (i.e., different from the ones used in the training) or are safety-agnostic.

\section{CONCLUSIONS}\label{Section5}
This paper studied MPGs  and MARL. MPGs have appealing properties that lead to the guaranteed performance of the MARL, including guaranteed pure NE existence, gradient play algorithm convergence, and attainability of the NE. We developed sufficient conditions for the MPG construction and proved that if the reward function and the MDP transition probability satisfy certain conditions, then the resulting MG is an MPG. Numerical results with applications to autonomous driving were reported. We found that the proposed reward design can accommodate the vehicles' driving objective design in general traffic scenarios, demonstrating the practicality of the developed MPG framework. Evaluation results suggest that the learned NE from MARL can enable safe and efficient autonomous vehicles in intersection-crossing scenarios and that the MARL has better robustness performance compared to single-agent RL against various surrounding vehicles' driving policies. More comprehensive evaluations in diverse traffic scenarios will be performed in future studies.  

\addtolength{\textheight}{-12cm}   



\bibliographystyle{IEEEtran}  
\bibliography{Reference}


\end{document}